\documentclass[11pt]{article}
\usepackage[margin=1in]{geometry}
\usepackage[tiny,compact]{titlesec}
\usepackage{lipsum}

\usepackage{amsmath}
\usepackage{subfigure}
\usepackage{epstopdf}
\usepackage{enumerate}
\usepackage{float}
\usepackage{amsfonts}
\usepackage{graphics}
\usepackage{caption}
\usepackage{multirow}
\usepackage{graphicx}
\usepackage{microtype}
\usepackage{algorithm}
\usepackage{algorithmic}
\newcommand{\citep}{\cite}
\usepackage{comment}
\usepackage{xspace}
\usepackage{url}
\usepackage[notheorems,nomathtools]{dgleich-math}
\usepackage{paralist}

\newcommand{\Vs}{\mathcal{V}}
\newcommand{\Gs}{\mathcal{G}}
\newcommand{\tGs}{\widetilde{\mathcal{G}}}
\newcommand{\tGst}{{\widetilde{\mathcal{G}}}^{(l)}}
\newcommand{\Es}{\mathcal{E}}
\newcommand{\tVs}{\widetilde{\mathcal{V}}}
\newcommand{\tEs}{\widetilde{\mathcal{E}}}

\newcommand{\st}{s\text{-}t}
\newcommand{\mincut}{min\text{-}cut\xspace}

\newcommand{\xtm}{\x^{(l-1)}}
\newcommand{\xt}{\x^{(l)}}
\newcommand{\xT}{\x^{(T)}}
\newcommand{\zt}{\z^{(l)}}
\newcommand{\wt}{\w^{(l)}}

\newcommand{\Wt}{\W^{(l)}}
\newcommand{\Lt}{\L^{(l)}}
\newcommand{\tLt}{\widetilde{\L}^{(l)}}
\newcommand{\tL}{\widetilde{\L}}
\newcommand{\tMt}{\widetilde{\M}^{(l)}}
\newcommand{\tLtp}{\widetilde{\L}^{(l+1)}}

\newcommand{\bt}{\b^{(l)}}
\newcommand{\btp}{\b^{(l+1)}}

\renewcommand{\vt}{\v^{(l)}}

\newcommand{\bd}{{\bf d}}

\newcommand{\Ss}{\mathcal{S}}
\newcommand{\Sscomp}{\bar{\mathcal{S}}}

\newcommand{\Ts}{\mathcal{T}}
\newcommand{\Seps}{{S_\epsilon}}
\newcommand{\vol}{\text{vol}}
\newcommand{\cut}{\text{cut}}
\newcommand{\tvol}{\widetilde{\text{vol}}}
\newcommand{\algname}{PIRMCut\xspace} % name of our algorithm
\newcommand{\bksolver}{{B-K Solver}\xspace} % name of our algorithm
\newcommand{\Tr}{^{\mathrm{T}}}

\renewcommand{\vec}{{\rm vec}}

\renewcommand{\b}{{\bf b}}

\newcommand{\e}{{\bf e}}
\newcommand{\f}{{\bf f}}
\newcommand{\g}{{\bf g}}

\renewcommand{\p}{{\bf p}}

\renewcommand{\r}{{\bf r}}

\renewcommand{\v}{{\bf v}}
\newcommand{\w}{{\bf w}}
\newcommand{\x}{{\bf x}}
\newcommand{\y}{{\bf y}}
\newcommand{\z}{{\bf z}}

\newcommand{\B}{{\bf B}}
\newcommand{\C}{{\bf C}}
\newcommand{\D}{{\bf D}}
\renewcommand{\E}{{\bf E}}

\newcommand{\I}{{\bf I}}

\renewcommand{\L}{{\bf L}}
\newcommand{\M}{{\bf M}}
\renewcommand{\O}{{\bf O}}
\renewcommand{\P}{{\bf P}}

\newcommand{\W}{{\bf W}}

\newcommand{\Z}{{\bf Z}}

\newcommand{\blambda}{\boldsymbol{\lambda}}

\newcommand{\ben}{\begin{enumerate}}
\newcommand{\een}{\end{enumerate}}

\newtheorem{theorem}{Theorem}[section]

\newtheorem{proposition}[theorem]{Proposition}

\newenvironment{proof}[1][Proof]{\begin{trivlist}
\item[\hskip \labelsep {\bfseries #1}]}{\end{trivlist}}

\newcommand{\qed}{\nobreak \ifvmode \relax \else
      \ifdim\lastskip<1.5em \hskip-\lastskip
      \hskip1.5em plus0em minus0.5em \fi \nobreak
      \vrule height0.75em width0.5em depth0.25em\fi}

\begin{document}

\title{\Large A Parallel Min-Cut Algorithm using Iteratively Reweighted Least Squares\thanks{Supported by NSF CAREER award CCF-1149756 and ARO Award 1111691-CNS}}

\author{
  Yao Zhu\\
  \texttt{yaozhu@purdue.edu}\\
  Purdue University
  \and
  David F.~Gleich\\
  \texttt{dgleich@purdue.edu}\\
  Purdue University
}

\date{}

\maketitle

%\pagenumbering{arabic}
%\setcounter{page}{1}%Leave this line commented out.

\begin{abstract} \small
We present a parallel algorithm for the undirected $\st$ \mincut problem with floating-point valued weights. Our overarching algorithm uses an iteratively reweighted least squares framework. This generates a sequence of Laplacian linear systems, which we solve using parallel matrix algorithms. Our overall implementation is up to 30-times faster than a serial solver when using 128 cores.
\end{abstract}

\section{Introduction}
We consider the undirected $\st$ \mincut problem. Our goal is a practical, scalable, parallel algorithm for problems with hundreds of millions or billions of edges and with floating point weights on the edges. Additionally, we expect there to be a sequence of such $\st$ \mincut computations where the difference between successive problems is small.
%only the weights on the edges changes. 
The motivation for such problems arises from a few recent applications including the FlowImprove method to improve a graph partition \citep{AndersenL:2008} and the GraphCut method to segment high-resolution images and MRI scans \citep{BoykovF:2006,HernandoKHZ:2010}. Both of these applications are limited by the speed of current $\st$ \mincut solvers that can handle problems with floating point weights. %and we seek to improve their speed.
We seek to accelerate such $\st$ \mincut computations.

For the undirected $\st$ \mincut problem, we present a Parallel Iteratively Reweighted least squares Min-Cut solver, which we call \algname\ for convenience. This algorithm draws its inspiration from the recent theoretical work on using Laplacians and electrical flows to solve max-flow/\mincut in undirected graphs \citep{ChristianoKMST:2011,KelnerLOS:2014,LeeRS:2013}. However, our exposition and derivation is entirely self-contained. There are three essential ingredients to our approach. The first essential ingredient is a variational representation of the $\ell_1$-minimization formulation of the undirected $\st$ \mincut (Section \ref{subsection:irls-algorithm}). This representation allows us to use the iteratively reweighted least squares (IRLS) method to generate a sequence of symmetric diagonally dominant linear systems whose solutions converge to an approximate solution (Theorem \ref{theorem:sequence-converge}).
%to the solution of the min-cut problem (Theorem 2.3).
We show that these systems are equivalent to electrical flows computation (Proposition \ref{proposition:wls-equal-electrical-flow}).
% and satisfy a Cheeger-type inequality with respect to the true cut value (Theorem 2.4).
We also prove a Cheeger-type inequality that relates an undirected $\st$ \mincut to a generalized eigenvalue problem (Theorem \ref{theorem:cheeger-type-inequality}).
The second essential ingredient is a parallel implementation of the IRLS algorithm using a parallel linear system solver. %on a distributed memory platform. 
The third essential ingredient is a two-level rounding procedure that uses information from the electrical flow solution to generate a smaller $\st$ \mincut problem suitable for sequential $\st$ \mincut solvers.

We have implemented \algname on a distributed memory platform and evaluated its performance on a set of test problems from road networks and MRI scans. 
Our solver, \algname, is 30 times faster (using 128 cores) than a state-of-the-art sequential $\st$ \mincut solver on a large road network.
%the fastest single-core code we could find to solve it. For solving a sequence of related linear systems, we 
In the experimental results, we also demonstrate the benefit of using warm starts when solving a sequence of related linear systems. We further show the advantage of the proposed two-level rounding procedure over the standard sweep cut in producing better approximate solutions.
%we further evaluate the improvement when using warm starts in the sequence of linear systems and describe the advantages of the two-level rounding heuristic.

At the moment, we do not have a precise and graph-size based runtime bound on \algname. We also acknowledge that, like most numerical solvers, it is only up to $\delta$-accurate. The focus of this initial manuscript is investigating and documenting a set of techniques that are principled and could lead to practically fast solutions on real world $\st$ \mincut problems. We compare our approach with those of others and discuss some further opportunities of our approach in the related work and discussions (Sections~\ref{section:related},~\ref{section:discussions}).
\section{An IRLS algorithm for undirected $\st$ \mincut}
In this section, we describe the derivation of the IRLS algorithm for the undirected $\st$ \mincut problem. We first introduce our notations. Let $\Gs=(\Vs,\Es)$ be a weighted, undirected graph. Let $n=|\Vs|$, and $m=|\Es|$. We require for each undirected edge $\{u,v\}\in \Es$, its weight $c(\{u,v\})>0$. Let $s$ and $t$ be two distinguished nodes in $\Gs$, and we call $s$ the \textit{source node} and $t$ the \textit{sink node}. The problem of undirected $\st$ \mincut is to find a partition of $\Vs=\Ss\cup \Sscomp$ with $s\in\Ss$ and $t\in \Sscomp$ such that the cut value
\begin{align*}
\text{cut}(\Ss,\Sscomp) &= \sum_{\{u,v\}\in\Es,u\in\Ss,v\in\Sscomp}c(\{u,v\})
\end{align*}
is minimized. For the interest of solving the undirected $\st$ \mincut problem, we assume $\Gs$ to be connected. We call the subgraph of $\Gs$ induced on $\Vs\backslash\{s,t\}$ the \textit{non-terminal graph}, and denote it by $\tGs=(\tVs,\tEs)$. We call the edges incident to $s$ or $t$ the \textit{terminal edges}, and denote them by $\Es^{\Ts}$.

%\subsection{Undirected $\st$ \mincut as $\ell_1$-minimization}
The undirected $\st$ \mincut problem can be formulated as an $\ell_1$-minimization problem. Let $\B\in\{-1,0,1\}^{m\times n}$ be the edge-node incidence matrix corresponding to an arbitrary orientation of $\Gs$'s edges, and $\C$ be the $m\times m$ diagonal matrix with $c(\{u,v\})$ on the main diagonal. Further let $\f=[1\ 0]^{\Tr}$, and $\Phi^{\Tr}=[\e_s\ \e_t]$ where $\e_s$ ($\e_t$) is the $s$\text{-}th ($t$\text{-}th) standard basis, then  the undirected $\st$ \mincut problem is
\begin{equation} \label{eq:l1-min}
 \MINone{\x}{ \normof[1]{\C \B \x} }{\Phi \x = \f, \quad \x \in [0,1]^n. }
\end{equation}

% \begin{align}
% \label{eq:l1-min}
% \min_{\x} & \ \ \ \ \|\C\B\x\|_1\\
% \nonumber
% \mbox{s.t.} & \ \ \ \ \Phi\x=\f\\
% \nonumber
% & \ \ \ \ \x\in[0,1]^n.
% \end{align}

\subsection{The IRLS algorithm}
\label{subsection:irls-algorithm}
The IRLS algorithm for solving the $\ell_1$-minimization problem (\ref{eq:l1-min}) is motivated by the variational representation of the $\ell_1$ norm
\begin{align*}
\|\z\|_1 &= \textstyle \min_{\w\geq 0}\frac{1}{2}\sum_{i}\left(\frac{\z_i^2}{\w_i}+\w_i\right)
\end{align*}
Using this variational representation, we can rewrite the $\ell_1$-minimization problem (\ref{eq:l1-min}) as the following joint minimization problem
\begin{eqnarray}\label{eq:joint-min}
 & \MINof[\x,\w \ge 0]  & H(\x,\w)=\frac{1}{2}\sum_{i}\left(\frac{(\C\B\x)_i^2}{\w_i}+\w_i\right) \\
 & \subjectto  & \Phi \x = \f, \quad \x \in [0,1]^n \label{eq:boundary-constraint}
 %\MINone{\x,\w \ge 0}{ H(\x,\w)=\frac{1}{2}\sum_{i}\left(\frac{(\C\B\x)_i^2}{\w_i}+\w_i\right) } {\Phi \x = \f, \quad \x \in [0,1]^n. }
\end{eqnarray}
% \begin{align}
% \label{eq:joint-min}
% \min_{\x,\w\geq 0} & H(\x,\w)=\frac{1}{2}\sum_{i}\left(\frac{(\C\B\x)_i^2}{\w_i}+\w_i\right)\\
% \label{eq:boundary-constraint}
% \mbox{s.t.} & \ \ \ \ \Phi\x=\f\\
% \label{eq:interval-constraint}
% & \ \ \ \ \x\in[0,1]^n.
% \end{align}
The IRLS algorithm for (\ref{eq:l1-min}) can be derived by applying alternating minimization to (\ref{eq:joint-min}) \citep{Beck:2013}. Given an IRLS iterate $\xtm$ satisfying constraint (\ref{eq:boundary-constraint}), the next IRLS iterate $\xt$ is defined according to the following two alternating steps:
\begin{itemize}
\item Step 1. Compute the \textit{reweighting vector} $\wt$, of which each component is defined by
\begin{align}
\label{eq:define-reweight-vector}
\wt_i=\sqrt{(\C\B\xtm)_i^2+\epsilon^2}
\end{align}
where $\epsilon>0$ is a smoothing parameter that guards against the case of $(\C\B\xtm)_i=0$.
\item Step 2. Let $\Wt=\diag(\wt)$, update $\xt$ by solving the weighted least squares (WLS) problem
\begin{equation} \label{eq:constrained-wls}
 \MINof[\x] \quad \frac{1}{2}\x^{\Tr}\B^{\Tr}\C(\Wt)^{-1}\C\B\x \quad \subjectto \quad \Phi\x = \f.
\end{equation}
% \begin{align}
% \label{eq:constrained-wls}
% \min_{\x} & \ \ \ \ \frac{1}{2}\x^{\Tr}\B^{\Tr}\C(\Wt)^{-1}\C\B\x\\
% \nonumber
% \mbox{s.t.} & \ \ \ \ \Phi\x=\f.
% \end{align}
\end{itemize}

\begin{comment}
%%%%%%12-03-2014. This paragraph is commented out so that the Appendix portion is merged in.
The solution $\xt$ is always well-defined because the linear system \eqref{eq:constrained-wls} is non-singular when $\Gs$ is connected. The iterate also satisfies $\xt\in [0,1]^n$. These proofs are fairly standard and are deferred to the appendix (Section~\ref{sec:well-defined}).
\end{comment}

%%%%%%12-03-2014---Begin portion of Appendix.
We prove that the sequence of iterates from the IRLS algorithm is well-defined, that is, each iterate $\xt$ exists and satisfies $\xt\in [0,1]^n$. Solving (\ref{eq:constrained-wls}) leads to solving the saddle point problem
\begin{align}
\label{eq:saddle-point-system}
\left[\begin{array}{ll}
\Lt & \Phi^{\Tr}\\
\Phi & \O
\end{array}
\right]
\left[\begin{array}{l}
\x\\
\blambda
\end{array}
\right]&=
\left[\begin{array}{l}
0\\
\f
\end{array}\right]
\end{align}
Regarding the nonsingularity of the linear system (\ref{eq:saddle-point-system}), we have the following proposition.
\begin{proposition}
A sufficient and necessary condition for the linear system (\ref{eq:saddle-point-system}) to be nonsingular is that $\Gs$ is connected.
\label{prop:nonsingular}
\end{proposition}
\begin{proof}
When $\Gs$ is connected, $\ker(\Lt)=\mbox{span}\{\e\}$, where $\e$ is the all-one vector. From the definition of $\Phi$, we know $\e\not\in\ker(\Phi)$, i.e., $\ker(\Lt)\cap\ker(\Phi)=\{0\}$. With $\Phi$ being full rank, we apply Theorem 3.2 from \citep{Benzi:2005} to prove the proposition.
\qed
\end{proof}
We now prove that the solution $\xt$ to (\ref{eq:saddle-point-system}) automatically satisfies the interval constraint.
\begin{proposition}
\label{prop:interval-constraint-satisfaction}
If $\Gs$ is connected, then for each IRLS update we have $\xt\in [0,1]^n$.
\end{proposition}
\begin{proof}
According to proposition \ref{prop:nonsingular}, the linear system (\ref{eq:saddle-point-system}) is nonsingular when $\Gs$ is connected. We apply the null space method \citep{Benzi:2005} to solve it. Let $\Z$ be the matrix from deleting the $s,t$ columns of the identity matrix $\I_n$, then $\Z$ is a null basis of $\Phi$. Because $\Phi\e_s=\f$, using $\Z$ and $\e_s$ we reduce (\ref{eq:saddle-point-system}) to the following linear system
\begin{align}
(\Z^{\Tr}\Lt\Z)\v &= -\Z^{\Tr}\Lt\e_s
\label{eq:reduced-laplacian-system}
\end{align}
where $\v$ is the vector by deleting the $s$ and $t$ components from $\x$. Note $\bt=-\Z^{\Tr}\Lt\e_s\geq 0$ because $\Lt$ is a Laplacian. We call $\tLt=\Z^{\Tr}\Lt\Z$ the \textit{reduced Laplacian}, and $(\tLt)^{-1}\geq 0$ because it's a Stieltjes matrix (see Corollary 3.24 of \citep{Varga:2000}). Thus $\vt=(\tLt)^{-1}\bt\geq 0$. Also note that $\tLt\e\ge \bt$, thus $\tLt(\e-\vt)=\tLt\e-\bt\ge 0$. Using $(\tLt)^{-1}\ge 0$ again, we have $\e-\vt\ge 0$. In summary, we have $0\le \vt\le \e$.
\qed
\end{proof}
%%%%%%12-03-2014---End portion of Appendix.

Now define
\begin{align}
\label{eq:reweighted-laplacian}
\Lt&=\B^{\Tr}\C(\Wt)^{-1}\C\B
\end{align}
we note that $\Lt$ is a \textit{reweighted Laplacian} of $\Gs$, where the edge weights are reweighted by the diagonal matrix $(\Wt)^{-1}$. We also define the \emph{reduced Laplacian} to be $\tLt=\Z^{\Tr}\Lt\Z$, in which $\Z$ is the matrix from deleting the $s,t$ columns of the identity matrix $\I_n$.
% is the reweighted Laplacian of $\Gs$ with the rows and columns corresponding to nodes $s$ and $t$ deleted. Let $\Z$ be the matrix from deleting the $s,t$ columns of the identity matrix $\I_n$, then the reduced Laplacian is .

In applying the IRLS algorithm to solve the undirected $\st$ \mincut problem, we start with an initial vector $\x^{(0)}$ satisfying (\ref{eq:boundary-constraint}). We take $\x^{(0)}$ to be the solution to (\ref{eq:constrained-wls}) with $\W^{(0)}=\C$. Then we generate the IRLS iterates $\xt$ for $l=1,\ldots,T$ using (\ref{eq:define-reweight-vector}) and (\ref{eq:constrained-wls}) alternatingly. 
Finally we apply a rounding procedure on $\xT$ to get a binary cut indicator vector. We discuss the details of the rounding procedure in Section \ref{subsection:two-level-rounding-procedure}. Essentially what the IRLS algorithm does is solving a sequence of reduced Laplacian systems for $l=1,\ldots,T$. It is already known that undirected max-flow/\mincut can be approximated by solving a sequence of electrical flow problems \citep{ChristianoKMST:2011}. We make the connection between the IRLS and the electrical flow explicit by proposition~\ref{proposition:wls-equal-electrical-flow}. Because of this connection between IRLS and the electrical flow, we also call $\xt$ the \textit{voltage vector} in the following.
\begin{proposition}
The WLS (\ref{eq:constrained-wls}) solves an electrical flow problem. Its flow value is $(\xt)^{\Tr}\Lt\xt$.
\label{proposition:wls-equal-electrical-flow}
\end{proposition}

\begin{comment}
%%%%%%12-03-2014. This paragraph is commented out so that the Appendix portion is merged in.
Please see the appendix for the proof (Section~\ref{sec:irls-electrical-flow}). 
\end{comment}

%%%%%%12-03-2014---Begin portion of Appendix.
\begin{proof}
Let $\xt$ be the solution to (\ref{eq:constrained-wls}) and $\zt=\C(\Wt)^{-1}\C\B\xt$. From $\frac{1}{2}(\xt)^{\Tr}\Lt\xt=\frac{1}{2}(\zt)^{\Tr}\C^{-1}\Wt\C^{-1}\zt$, $\zt$ minimizes the energy function $E(\z) = \frac{1}{2}\z^{\Tr}\C^{-1}\Wt\C^{-1}\z$. We now prove that $\zt$ satisfies the flow conservation constraint.
\begin{align*}
\B^{\Tr}\zt &= \B^{\Tr}\C(\Wt)^{-1}\C\B\xt=\Lt\xt\\
&=-\Phi^{\Tr}\blambda
\end{align*}
where the last equality is by the first equation in (\ref{eq:saddle-point-system}).
From $\Phi^{\Tr}=[\e_s\ \e_t]$, we have $(\B^{\Tr}\zt)_u=0$ when $u\not=s,t$. Also note that $\e^{\Tr}\B^{\Tr}\zt=0$. Thus $\zt$ is an electrical flow. From $\xt_s=1$ and $\xt_t=0$, the flow value of $\zt$ is given by
\begin{align*}
\mu(\zt) &= (\xt)^{\Tr}\B^{\Tr}\zt=(\xt)^{\Tr}\B^{\Tr}\C(\Wt)^{-1}\C\B\xt\\
&=(\xt)^{\Tr}\Lt\xt \qed
\end{align*}
\end{proof}
%%%%%%12-03-2014---End portion of Appendix.

\subsection{Convergence of the IRLS algorithm}
We now apply a general proof of convergence from Beck~\cite{Beck:2013} to our case. 
Because of the smoothing parameter $\epsilon$ introduced in defining the reweighting vector $\wt$ as in (\ref{eq:define-reweight-vector}), the IRLS algorithm is effectively minimizing a smoothed version of the $\ell_1$-minimization problem (\ref{eq:l1-min}) defined by
\begin{equation} \label{eq:l1-min-smooth-version}
\MINone{\x}{ \Seps(\x)=\sum_{i=1}^m\sqrt{(\C\B\x)_i^2+\epsilon^2} }{\Phi\x=\f, \quad \x\in[0,1]^n.}
\end{equation}

% \begin{align}
% \label{eq:l1-min-smooth-version}
% \min_{\x} & \ \ \ \ \Seps(\x)=\sum_{i=1}^m\sqrt{(\C\B\x)_i^2+\epsilon^2}\\
% \nonumber
% \mbox{s.t.} & \ \ \ \ \Phi\x=\f, \x\in[0,1]^n
% \nonumber
% \end{align}
%Beck's three theorems show that the objective value converges with two different types of bounds, and furthermore, establishes the convergence properties of the iterates themselves. 
Using results from Beck~\cite{Beck:2013}, we have the following nonasymptotic and asymptotic convergence rates of the IRLS algorithm.
\begin{theorem}[Theorem 4.1 of \citep{Beck:2013}]
Let $\{\xt\}_{l\ge 0}$ be the sequence generated by the IRLS method with smoothing parameter $\epsilon$. Then for any $T\ge 2$
\small
\begin{align*}
\Seps(\xT)-\Seps^{\ast}&\leq \max\left\{\left(\frac{1}{2}\right)^{\frac{T-1}{2}}(\Seps(\x^{(0)})-\Seps^{\ast}), \frac{8\lambda_{\max}(\B^{\Tr}\C^2\B)R^2}{\epsilon(T-1)}\right\}
\end{align*}
\normalsize
where $\Seps^{\ast}$ is the optimal value of (\ref{eq:l1-min-smooth-version}), and $R$ is the diameter of the level set of (\ref{eq:joint-min}) (see (3.7) in \citep{Beck:2013}).
\begin{comment}
{\small
\begin{align*}
R&=\max_{(\x,\w)}\max_{(\x^{\ast},\w^{\ast})\in X^\ast}\{\|\left(\begin{array}{c}
\x\\
\w
\end{array}\right)-
\left(\begin{array}{c}
\x^{\ast}\\
\w^{\ast}
\end{array}\right)\|:H(\x,\w)\leq H(\x^0,\w^0)\}
\end{align*}
}
with $X^\ast$ being the optimal set of the joint minimization problem (\ref{eq:joint-min}).
\end{comment}
\end{theorem}
\begin{theorem}[Theorem 4.2 of \citep{Beck:2013}]
Let $\{\xt\}_{l\ge 0}$ be the sequence generated by the IRLS method with smoothing parameter $\epsilon$. Then there exists $K>0$ such that for all $T\ge K+1$
\begin{align*}
\Seps(\xT)-\Seps^{\ast}&\leq\frac{48R^2}{\epsilon(T-K)}
\end{align*}
\end{theorem}
Regarding the convergence property of the iterates $\{\xt\}_{l\ge 0}$, we have the following theorem.
\begin{theorem}[Lemma 4.3 of \citep{Beck:2013}]
Let $\{\xt\}_{l\ge 0}$ be the sequence generated by the IRLS method with smoothing parameter $\epsilon$. Then
\begin{enumerate}[(i)]
\item any accumulation point of $\{\xt\}_{l\ge 0}$ is an optimal solution of (\ref{eq:l1-min-smooth-version}).
\item for any $i\in\{1,\ldots,m\}$, the sequence $\{|(\C\B\xt)_i|\}_{l\ge 0}$ converges.
\end{enumerate}
\label{theorem:sequence-converge}
\end{theorem}

\subsection{A Cheeger-type inequality}
In Section \ref{subsection:irls-algorithm} we have shown that the IRLS algorithm relates the undirected $\st$ \mincut problem to solving a sequence of reduced Laplacian systems. Here we prove a Cheeger-type inequality that characterizes the undirected $\st$ \mincut by the second finite generalized eigenvalue of a particular matrix pencil. We start with some notations and definitions. Let $C=2\sum_{\{u,v\}\in\Es}c(\{u,v\})$ i.e., twice the total edge weights of $\Gs$. We define a weight function $\bd(\cdot)$ on $\Vs$ to be
\begin{align}
\label{eq:cheeger-node-weight}
\bd(u)&=\left\{\begin{array}{ll}
C & u=s \mbox{ or } u=t\\
0 & \mbox{otherwise}
\end{array}\right.
\end{align}
The volume of a subset $\Ss\subset \Vs$ is defined by $\vol(\Ss)=\sum_{u\in\Ss}\bd(u)$. As usual we define $\phi(\Ss)=\frac{\cut(\Ss,\Sscomp)}{\min\{\vol(\Ss),\vol(\Sscomp)\}}$, and the expansion parameter to be $\phi=\min_{\Ss\subset \Vs}\phi(\Ss)$. We have $\phi(S)<\infty$ if and only if $(\Ss,\Sscomp)$ is an $\st$ cut. And $\phi(S)=\phi$ if and only if $(\Ss,\Sscomp)$ is an $\st$ \mincut.
Let $\D$ be the $n\times n$ matrix with its main diagonal being $\bd$, and let $\L$ be the Laplacian of $\Gs$. We consider the generalized eigenvalue problem of the matrix pencil $(\L,\D)$
\begin{equation} \label{eq:spectral}
\L\x = \lambda \D \x
\end{equation}
Because $\D$ is of rank 2, the pencil $(\L,\D)$ has only two finite generalized eigenvalues (see Definition 3.3.1 in \citep{ToledoA:2012}). We denote them by $\lambda_1\le\lambda_2$. We state a Cheeger-type inequality for undirected $\st$ \mincut as follows.
\begin{theorem} Let $\phi$ be the min-cut and let $\lambda_2$ be the non-zero eigenvalue of \eqref{eq:spectral}, then
$\frac{\phi^2}{2}\leq \lambda_2 \leq 2\phi$.%
\label{theorem:cheeger-type-inequality}
\end{theorem}
We defer this proof to the appendix (Section~\ref{sec:cheeger-proof}) because it is a straightforward generalization of the proof of Theorem 1 in \citep{Chung:2007}. It also follows from a recent generalized Cheeger inequality by~\cite{Koutis:2014}.%entirely standard for a Cheeger-type inequality and follows the result from Theorem 1 in \citep{Chung:2007} closely. 

\section{Parallel implementation of the IRLS algorithm}
\label{section:parallel-implementation-irls}
In this section, we describe our parallel implementation of the IRLS algorithm. In addition, we also detail the rounding procedure we use to obtain an approximate $\st$ \mincut from the voltage vector $\xT$ output from the IRLS iterations.

\subsection{A parallel iterative solver for the sequence of reduced Laplacian systems}
\label{subsection:parallel-iterative-solver-the sequence}
The main focus of the IRLS algorithm is to solve a sequence of reduced Laplacian systems for $l=1,\ldots,T$. We keep $\tLt=\Z^{\Tr}\Lt\Z$ and set $\bt=-\Z^{\Tr}\Lt\e_s$ in the following; note that the system $\tLt \v = \bt$ is equivalent to \eqref{eq:constrained-wls} as explained in the proof to Proposition \ref{prop:interval-constraint-satisfaction}. %Appendix~\ref{sec:well-defined}.
Because $\tLt$ is symmetric positive definite, we can use the preconditioned conjugate gradient (PCG) algorithm to solve it in parallel~\citep{Saad:2003}. For preconditioning, we choose to use the block Jacobi preconditioner because of its high parallelism. Specifically, we extract from $\tLt$ a $p\times p$ block diagonal matrix $\tMt=\diag(\tMt_1,\ldots,\tMt_p)$, where $p$ is the number of processes. We use $\tMt$ to precondition solving the linear system of $\tLt$. In each preconditioning step of the PCG iteration, we need to solve a linear system involving $\tMt$. Because of the block diagonal structure of $\tMt$, the $p$ independent subsystems
\begin{align}
\label{eq:block-jacobi-subsystem}
\tMt_j \y_j&=\r_j\ j=1,\ldots,p
\end{align}
can be solved in parallel. We consider two strategies to solve the subsystems (\ref{eq:block-jacobi-subsystem}): 
(I) LU factorization to exactly solve (\ref{eq:block-jacobi-subsystem}); and (II) incomplete LU factorization with zero fill-in (ILU(0)) to approximately solve (\ref{eq:block-jacobi-subsystem}).
Given that the nonzero structure of $\tLt$ never changes (it's fully determined by the non-terminal graph $\tGs$), the symbolic factorization of LU and ILU(0) needs to be done only once. As we will demonstrate in Section \ref{section:experimental-results}, applying (I) or (II) highly depends on the application. Given that we are solving a sequence of related linear systems, where the solution $\vt$ is used to define $\tLtp$ and $\btp$, we adopt a \textit{warm start} heuristic. Specifically, we let $\vt$ be the initial guess to the parallel PCG iterations for solving the $(l+1)$-th linear system. In contrast, \textit{cold start} involves always using the zero initial guess.

\begin{comment}
%%%%%%12-03-2014. This paragraph is commented out so that the Appendix portion is merged in.
Additional details of our parallel scheme regarding data partitioning and parallel data layout are given in the appendix (Section~\ref{sec:parallel-partitioning} and \ref{sec:graph-distribution}). In brief, we use ParMETIS \citep{KarypisV:1998} for graph partitioning and then distribute the matrix in block rows. 
\end{comment}

%%%%%%12-03-2014---Begin portion of Appendix.
\subsection{Parallel graph partitioning} \label{sec:parallel-partitioning}
Given a fixed number of processes $p$, we consider the problem of extracting from $\tLt$ an effective and efficient block Jacobi preconditioner $\tMt$. For this purpose, we impose the following two objectives:
\begin{enumerate}[(i)]
\item The Frobenius norm ratio $\|\tMt\|_{\text{F}}/\|\tLt\|_{\text{F}}$ is as large as possible.
\item The sizes and number of nonzeros in $\tMt_j,\ j=1,\ldots,p$ are well balanced.
\end{enumerate}
The above two objectives can be formulated as a $k$-way graph partitioning problem \citep{KarypisV:1998,KarypisV:1999}. Let $\tGst$ be the reweighted non-terminal graph whose edge weights are determined by the off-diagonal entries of $\tLt$. Ideally we need to apply the $k$-way graph partitioning to $\tGst$ for each $l=1,\ldots,T$. However, according to (\ref{eq:reweighted-laplacian}) we observe that large edge weights tend to be large after reweighting because of the $\C^2$ in (\ref{eq:reweighted-laplacian}). This motivates us to reuse graph partitioning. Specifically, we use the parallel graph partitioning package ParMETIS\footnote{\url{http://glaros.dtc.umn.edu/gkhome/metis/parmetis/overview}} \citep{KarypisV:1998} to partition the non-terminal graph $\tGs$ into $p$ components with the objective of minimizing the weighted edge cut, while balancing the volumes of different components. We then use this graph partitioning result for all the following IRLS iterations. The graph partitioning result implies a reordering of the nodes in $\Vs\backslash\{s,t\}$ such that nodes belonging to the same component are numbered continuously. Let the permutation matrix of this reordering be $\P$. We then extract $\tMt$ as the block diagonal submatrix of the reordered matrix $\P\tLt\P^{\Tr}$.

\subsection{Graph distribution and parallel graph reweighting} \label{sec:graph-distribution}
The input graph $\Gs$ consists of the non-terminal graph $\tGs$ and the terminal edges $\Es^{\Ts}$. Let $\tL$ be the graph Laplacian of $\tGs$. For distributing $\tGs$ among $p$ processes, we partition the reordered matrix $\P\tL\P^{\Tr}$ into $p$ block rows and assign each block row to one process. The graph partitioning also induces a partition of the terminal edges as $\Es^{\Ts}=\bigcup\limits_{j=1}^p\Es^{\Ts}_j$, and each $\Es^{\Ts}_j$ is assigned to one process. Accordingly we also partition and distribute the permuted vectors $\P\xt$ and $\P\bt$ among $p$ processes. At the $l$-th IRLS iteration, we need to form the reweighted Laplacian (\ref{eq:reweighted-laplacian}) using the voltage vector $\P\xtm$. After the $p$ processes have communicated to acquire their needed components of $\P\xtm$, they can form their local block row of $\P\tLt\P^{\Tr}$ and $\P\bt$ in parallel. Thus the cost of one parallel graph reweighting is equivalent to that of one parallel matrix-vector multiplication. Note that the $k$-way graph partitioning usually results in a small number of edges between different components. This characteristic helps to reduce the process communication cost.
%%%%%%12-03-2014---End portion of Appendix.

\subsection{A two-level rounding procedure}
\label{subsection:two-level-rounding-procedure}
Let $\xT$ be the voltage vector output from running the IRLS algorithm for $T$ iterations. We consider the problem of converting $\xT$ to a binary cut indicator vector. A standard technique for this purpose is sweep cut \citep{Vishnoi:2012}, which is based on sorting the nodes according to $\xT$. Here we propose a heuristic rounding procedure, which we refer to as the \textit{two-level rounding procedure}. Empirically we observe the components of $\xt$ tend to converge to $\{0,1\}$ as $l$ gets larger. We call this effect the \textit{node voltage polarization}. This observation motivates the idea of coarsening the graph $\Gs$ using $\xT$. Given $\xT$ and another two parameters $\gamma_0$ and $\gamma_1$, we partition $\Vs=\Ss_0\cup\Ss_1\cup\Ss_c$ as follows:
\begin{align*}
u & \in \left\{\begin{array}{ll}
\Ss_0, & \mbox{ if } \xT_u\leq \gamma_0 \\
\Ss_1, & \mbox{ if } \xT_u\geq \gamma_1 \\
\Ss_c, & \mbox{ if } \gamma_0 < \xT_u < \gamma_1
\end{array}\right.
\end{align*}
Intuitively, we coalesce node $u$ with the sink (source) node if $\xT_u$ is small (large) enough, and otherwise we leave it alone. Contracting $\Ss_0$ ($\Ss_1$) to a single node $t_c$ ($s_c$), we define a coarsened graph $\Gs_c=(\Vs_c,\Es_c)$ with node set $\Vs_c=\{s_c,t_c\}\cup\Ss_c$. 

\begin{comment}
%%%%%%12-03-2014. This paragraph is commented out so that the Appendix portion is merged in.
See Appendix~\ref{sec:two-level-details} for a formal definition.
\end{comment}

%%%%%%12-03-2014---Begin portion of Appendix.
The weighted edge set $\Es_c$ of the coarsened graph is defined according to the following cases
\begin{compactitem}
\item $\{u,v\}\in\Es_c$ with weight $c(\{i,j\})$ if $u,v \in \Ss_c$ such that $\{u,v\}\in\Es$.
\item $\{s_c,u\}\in\Es_c$ with weight $\sum_{v\in \Ss_1}c(\{u,v\})$.
\item $\{t_c,u\}\in\Es_c$ with weight $\sum_{v\in \Ss_0}c(\{u,v\})$.
\item $\{s_c,t_c\}\in\Es_c$ with weight $\sum_{u\in \Ss_1}\sum_{v\in \Ss_0}c(\{u,v\})$.
\end{compactitem}
%%%%%%12-03-2014---End portion of Appendix.

An $\st$ cut on $\Gs_c$ induces an $\st$ cut on $\Gs$. The idea of the two-level rounding procedure is to solve the undirected $\st$ \mincut problem on the coarsened graph $\Gs_c$ using a combinatorial max-flow/min-cut solver, e.g., \citep{BoykovK:2004}; and then get the corresponding $\st$ cut on $\Gs$. Regarding the quality of such an acquired $\st$ cut on $\Gs$, we have the following proposition.
\begin{proposition}
Let $\Vs=\Ss_0\cup\Ss_1\cup\Ss_c$ be the partition generated by coarsening using $\xT$ and parameters $\gamma_0,\gamma_1$. If there is an $\st$ \mincut on $\Gs$ such that $\Ss_1$ is contained in the source side, and $\Ss_0$ is contained in the sink side, then an $\st$ \mincut on $\Gs_c$ induces an $\st$ \mincut on $\Gs$.
\end{proposition}
Choosing the values of $\gamma_0$ and $\gamma_1$ is critical to the success of the two-level rounding procedure. On the one hand, we want to set their values conservatively so that not many nodes in $\Ss_0$ and $\Ss_1$ are coalesced to the wrong side. On the other hand, we want to set their values aggressively so that the size of $\Gs_c$ is much smaller than $\Gs$. We find a good trade off is achieved by clustering the components of $\xT$. Let the two centers returned from $K$-means on $\xT$ (with initial guesses $0.1$ and $0.9$) be $c_0$ and $c_1$, then we let $\gamma_0=c_0+0.05$ and $\gamma_1=c_1-0.05$.

\begin{algorithm}
\caption{\algname$(\Gs,s,t,p,T)$}
\begin{algorithmic}[1]
\label{alg:parallel-irls-min-cut}
\REQUIRE
A weighted and undirected graph $\Gs=(\Vs,\Es)$, the source node $s$, the sink node $t$, the number of processes $p$, and the number of IRLS iterations $T$.
\ENSURE
An approximate $\st$ \mincut on $\Gs$.
\STATE Let $\tGs$ be the non-terminal graph and $\Es^{\Ts}$ be the terminal edges. Use ParMETIS to partition $\tGs$ into $p$ components. Reorder and distribute $\tGs$ and $\Es^{\Ts}$ to $p$ processes accordingly.
\STATE Initialize $\x^{(0)}$ to be the solution to (\ref{eq:constrained-wls}) with $\W^{(0)}=\C$ using the parallel PCG solver.
\FOR{$l=1,\ldots,T$}
\STATE Execute the parallel IRLS algorithm by alternating between parallel graph reweighting and solving the reduced Laplacian system $\tLt \v = \bt$ using the parallel PCG solver.
\ENDFOR
\STATE Gather the voltage vector $\xT$ to the root process.
\STATE The root process applies a rounding procedure on $\xT$ to get an approximate $\st$ \mincut on $\Gs$.
\end{algorithmic}
\end{algorithm}
We summarize the algorithm \algname in Algorithm \ref{alg:parallel-irls-min-cut}. Note the main sequential part of \algname is the rounding procedure (either sweep cut or our two-level rounding).
\section{Related work}
\label{section:related}
Approaches to $\st$ \mincut based on $\ell_1$-minimization are common. Bhusnurmath and Taylor \citep{BhusnurmathT:2008}, for example, transform the $\ell_1$-minimization problem to its linear programming (LP) formulation, and then solve the LP using an interior point method with logarithmic barrier functions. Similar to our IRLS approach, their interior point method also leads to solving a sequence of Laplacian systems. 
%Recently it has been found that solving the Laplacian systems could be used as a general algorithm design paradigm on graphs \citep{Teng:2010}. 
In that vein, \citep{ChristianoKMST:2011} develops the then theoretically fastest approximation algorithms for $\st$ max-flow/\mincut problems in undirected graphs. Their algorithm involved a sequence of electrical flow problems defined via a multiplicative weight update.
%Because our IRLS algorithm is also solving a sequence of electrical flow problems,
We are still trying to formalize the relation between our IRLS update and the multiplicative weight update used in \citep{ChristianoKMST:2011}. More recent progress of employing electrical flows to approximately solve the undirected $\st$ max-flow/\mincut problems to achieve better complexity bounds include \citep{KelnerLOS:2014, LeeRS:2013}. In contrast, the main purpose of this manuscript is to implement a parallel $\st$ \mincut solver using this Laplacian paradigm, and evaluate its performance on a parallel computing platform. Beyond the electrical flow paradigm, which always satisfies the flow conservation constraint but relaxes the capacity constraint, the fastest combinatorial algorithms, like push-relabel and its variants \cite{CherkasskyG:1997,Goldberg:2009} and pseudoflow \citep{Hochbaum:2008}, relaxes the flow conservation constraint but satisfying the capacity constraint. However, parallelizing these methods is difficult.
\par
%Regarding the Laplacian solvers, two state-of-the-art linear solvers for Laplacian and symmetric diagonally dominant linear systems are CMG \citep{KoutisMT:2011} and LAMG \citep{LivneB:2012}.
Recent advances on nearly linear time solver for Laplacian and symmetric diagonally dominant linear systems are \citep{KelnerOSZ:2013,KoutisMT:2011,LivneB:2012}.
We do not incorporate the current available implementations of \citep{KoutisMT:2011} and \citep{LivneB:2012} into \algname because they are sequential. Instead we choose the block Jacobi preconditioned PCG because of its high parallelism. We notice an idea similar to our two-level rounding procedure that is called \textit{flow-based rounding} has been proposed in \citep{Lang:2005}. This idea was later rigorously formalized as the FlowImprove algorithm \citep{AndersenL:2008}.

\section{Experimental results}
In this section, we describe experiments with \algname on several real world graphs. These experiments serve to empirically demonstrate the properties of \algname, and also evaluate its performance on large scale undirected $\st$ \mincut problems. In the following experiments, we only consider floating-point valued instances because most interesting applications of undirected $\st$ \mincut produce floating-point valued problems, e.g., FlowImprove \citep{AndersenL:2008}, image segmentation \citep{BoykovF:2006}, MRI analysis \citep{HernandoKHZ:2010}, and energy minimization in MRF \citep{KolmogorovZ:2004}. Thus, we do not include those state-of-the-art max-flow/min-cut solvers for integer valued problems in the following experiments, like \texttt{hipr-v3.4}\footnote{\url{http://www.igsystems.com/hipr/download.html}} \citep{CherkasskyG:1997} and \texttt{pseudo-max-3.23}\footnote{\url{http://riot.ieor.berkeley.edu/Applications/Pseudoflow}} \citep{Hochbaum:2008}. We would have liked to compare different parallel max-flow/min-cut solvers \citep{DelongB:2008,ShekhovtsovH:2013,StrandmarkK:2010} regarding their performance and parallel scalability, but a detailed, fair comparison is beyond the scope of this extended abstract.
\label{section:experimental-results}

\subsection{Data Sets}
\begin{table}
\caption{{All graphs are undirected and we report the size of the non-terminal graph.}}
\begin{center}
\vspace{-0.6cm}
\begin{tabular}{rll}
  \hline
  Graph & $|\tVs|$ & $|\tEs|$ \\
  \hline
  \texttt{usroads-48} & 126,146 & 161,950 \\
  \texttt{asia\_osm} & 11,950,757 & 12,711,603 \\
  \texttt{euro\_osm} & 50,912,018 & 54,054,660 \\
  \texttt{adhead.n26c100} & 12,582,912 & 163,577,856 \\
  \texttt{bone.n26c100} & 7,798,784 & 101,384,192 \\
  \texttt{liver.n26c100} & 4,161,600 & 54,100,800 \\
  \hline
 \end{tabular}
\end{center}
\label{table:data-set}
\end{table}
The real world graphs we use to evaluate \algname are from two sources: road networks and N-D grid graphs. The first three road networks in Table \ref{table:data-set} are from the University of Florida Sparse Matrix collection \citep{DavisH:2011}. From these road graphs, we generate their undirected $\st$ \mincut instances using the FlowImprove algorithm \citep{AndersenL:2008} starting from a geometric bisection. The other three N-D grid graphs in Table \ref{table:data-set} are the segmentation problem instances from the University of Western Ontario max-flow datasets\footnote{\url{http://vision.csd.uwo.ca/maxflow-data}} \citep{BoykovF:2006}, on which we convert the edge weights to be floating-point valued by adding uniform random numbers in $[0,1]$.

\subsection{The effect of warm starts}
\begin{figure}
  \centering 
  \includegraphics[width=0.5\linewidth]{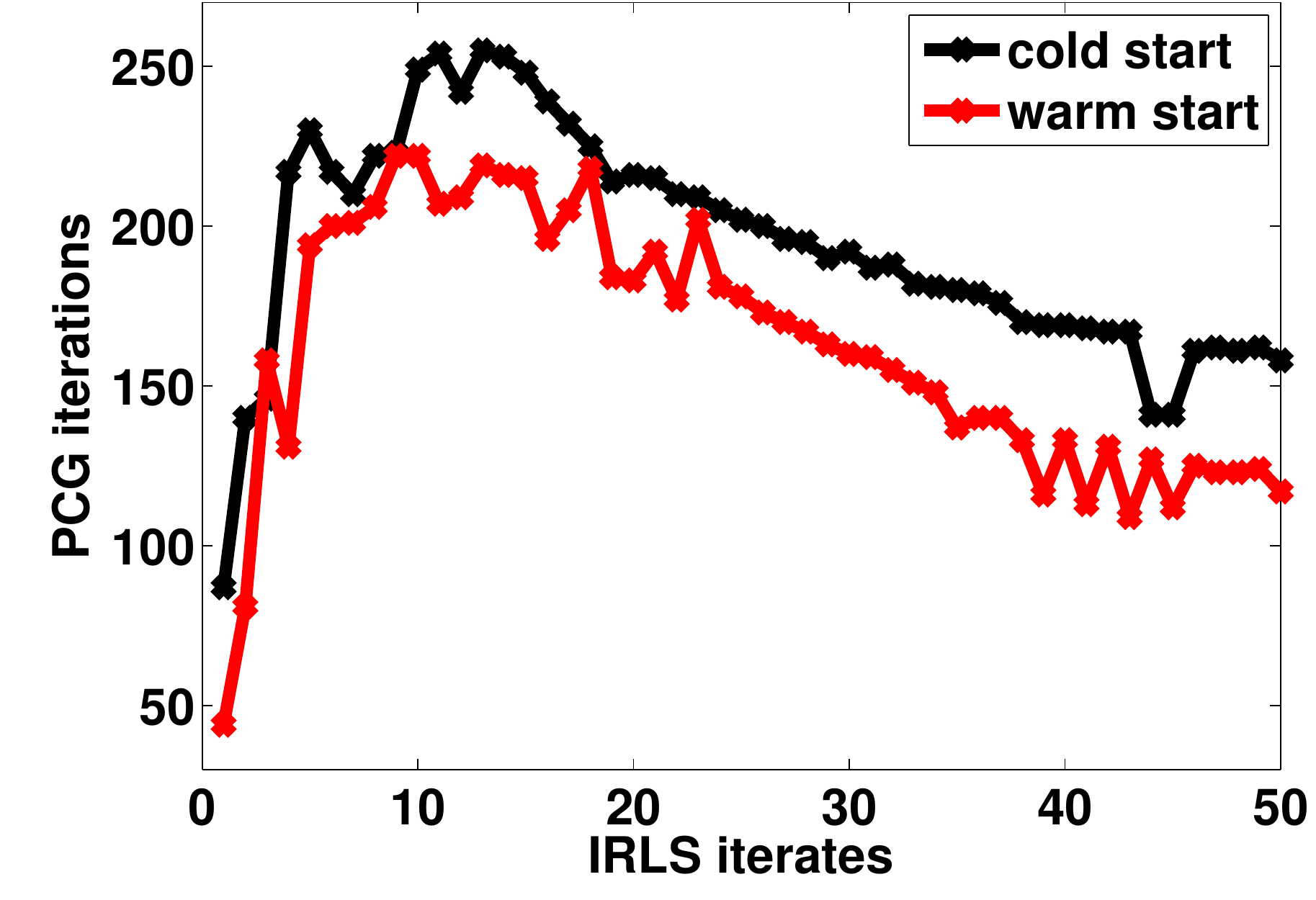}
  \caption{The number of PCG iterations to reach a residual of $10^{-3}$ is reduced by roughly 20\% by using warm starts for this sequence of Laplacian systems.}
  \label{figure:contrast-warm-cold}
\end{figure}
On the graph of \texttt{usroads-48}, we demonstrate the benefit of the warm start heuristic described in Section \ref{subsection:parallel-iterative-solver-the sequence}. We set the smoothing parameter $\epsilon=10^{-6}$ and run the IRLS algorithm for 50 iterations on 4 MPI processes. For each IRLS iterate, we set the maximum number of PCG iterations to be 300, and the stopping criterion in relative residual to be $10^{-3}$. In Figure \ref{figure:contrast-warm-cold} we plot the number of PCG iterations of using warm starts and cold starts. It is apparent that for most IRLS iterates, warm starts help to reduce the number of needed PCG iterations significantly, especially for later IRLS iterates. Another interesting phenomenon we observe from Figure \ref{figure:contrast-warm-cold} is that the difficulty of solving the reduced Laplacian system increases dramatically during the first several IRLS iterates, and then decreases later on. A possible explanation is that the IRLS algorithm makes faster progress during the early iterates.

\subsection{Effect of node voltage polarization}
Continuing with \texttt{usroads-48} as an example, we demonstrate node voltage polarization, which motivates the idea of two-level rounding procedure (see Section \ref{subsection:two-level-rounding-procedure}). We plot a heatmap of $\xt$ (after sorting its components) for $l=0,\ldots,50$ in Figure \ref{figure:node-voltage-polarization}. It is apparent that the polarization gets emphasized as $l$ becomes larger. We also see from Figure \ref{figure:node-voltage-polarization} that the values of $\xt$ change quickly for $l\leq 10$, which reinforces the observation of the IRLS algorithm making faster progress early.

\begin{figure}[ht!]
  \centering 
  \includegraphics[width=0.5\hsize]{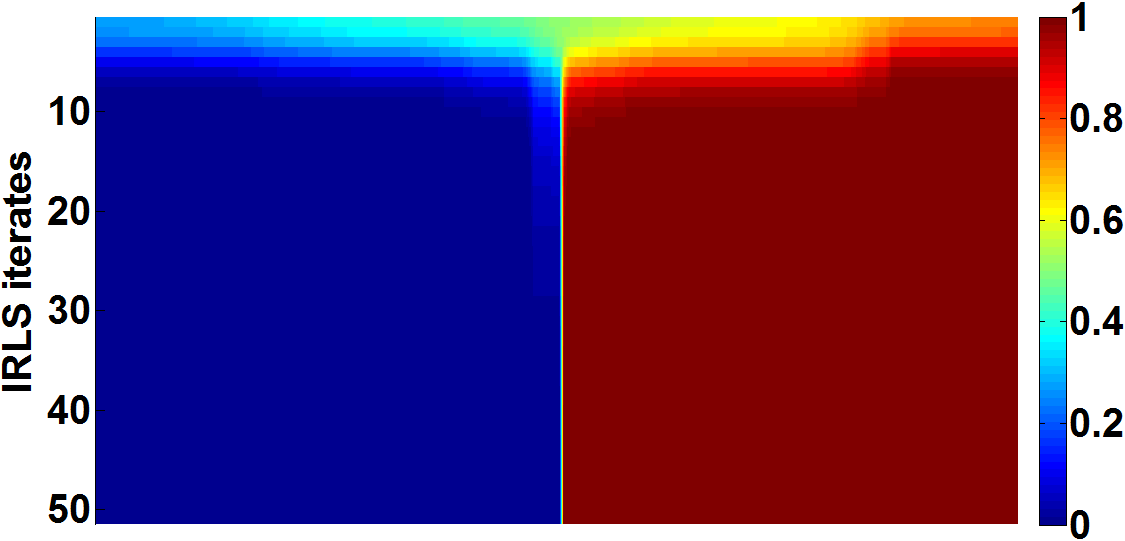}
  \caption{{Node voltage polarization. Each row is a sorted voltage vector $\xt$.}}
  \label{figure:node-voltage-polarization}
\end{figure}

\subsection{Performance evaluation on large graphs}
We evaluate the empirical performance of \algname on the five largest graphs in Table \ref{table:data-set}. Our parallel implementation of \algname is written in C++, and it is purely based on MPI, no multi-threading is exploited. In particular, the parallel PCG solver with block Jacobi preconditioner is implemented using the PETSc\footnote{\url{http://www.mcs.anl.gov/petsc/}} package \citep{Smith:2011}. All the experiments are conducted on a cluster machine with a total number of 192 Intel Xeon E5-4617 processors (8 nodes each with 24 cores). All the results reported in the following are generated using: smoothing parameter $\epsilon=10^{-6}$, number of IRLS iterations $T=50$, warm start heuristic, and 50 PCG iterations at maximum. And on the road networks, we use exact LU factorization by UMFPACK \citep{Davis:2004} to solve the subsystems (\ref{eq:block-jacobi-subsystem}), while on the N-D grid graphs we use ILU(0) in PETSc \citep{Smith:2011}. For implementing the two-level rounding procedure in \algname we use the Boykov-Kolmogorov solver\footnote{\url{http://pub.ist.ac.at/~vnk/software.html}} \citep{BoykovK:2004}, which is a state-of-the-art max-flow solver that can handle floating-point edge weights. We refer to it as the \bksolver in the following.

\begin{figure}[ht!]
  \centering 
  \includegraphics[width=0.5\hsize]{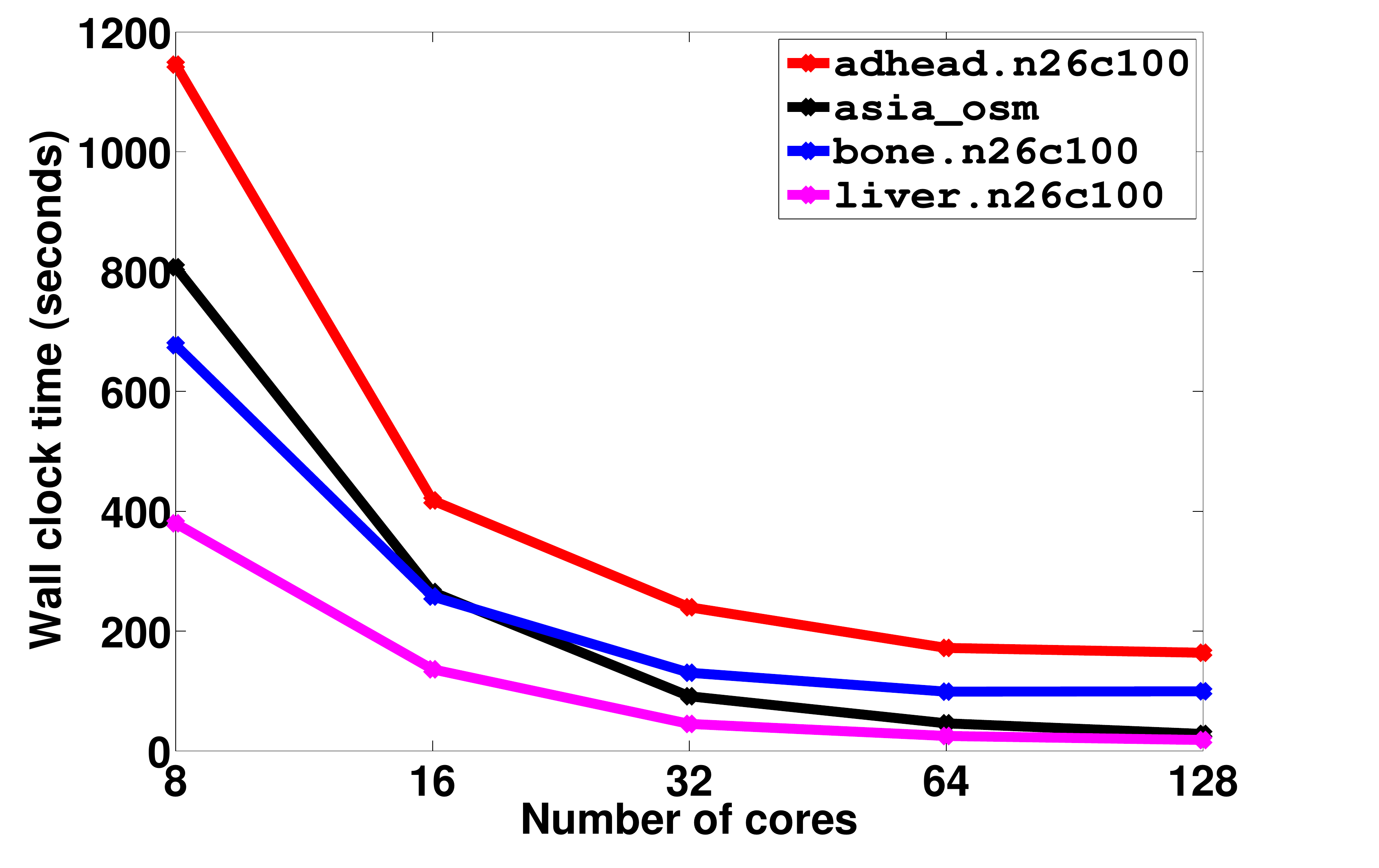}%
  \includegraphics[width=0.47\hsize]{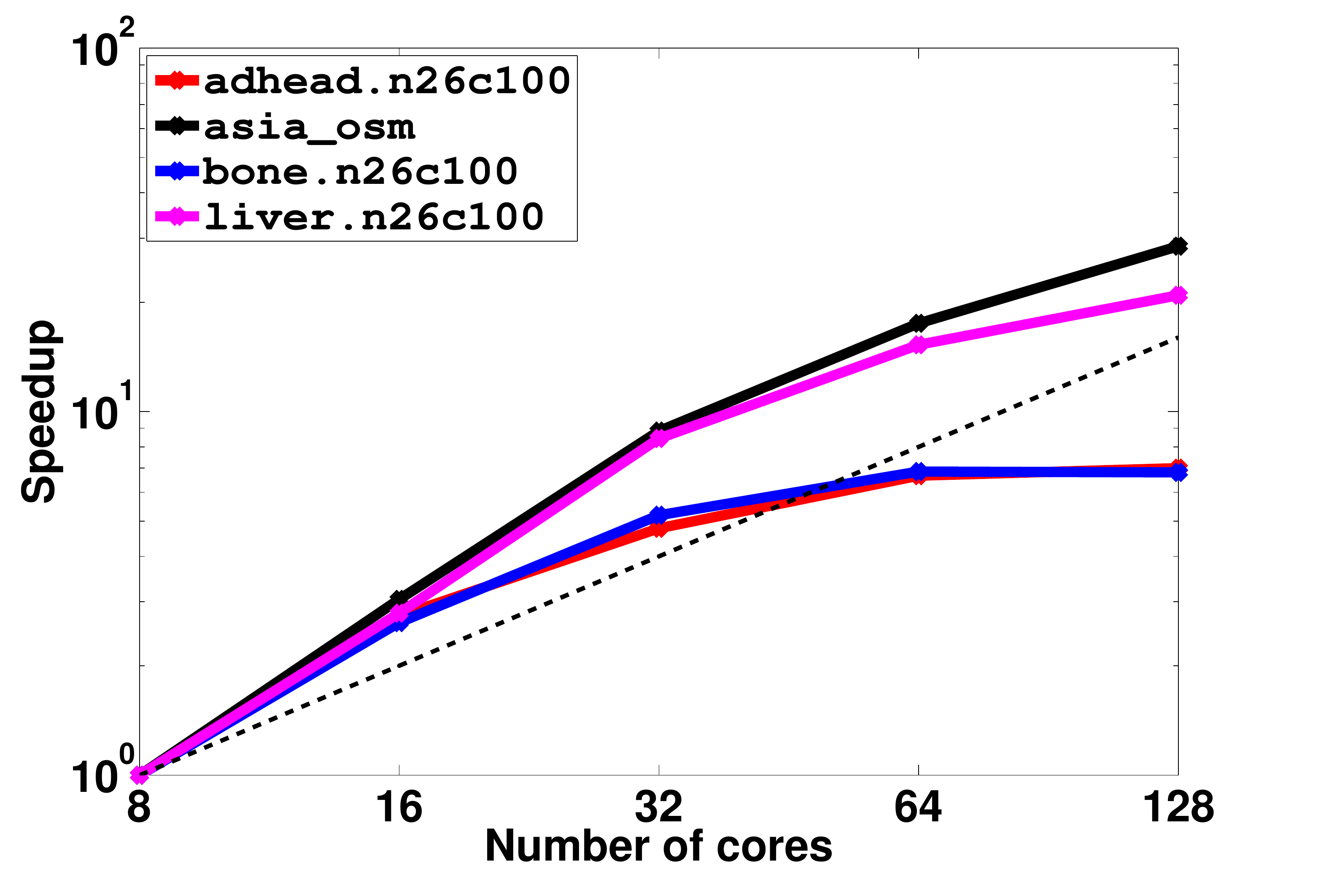}
  \caption{{Parallel scalability of the IRLS iterations of \algname on four large graphs.}}
  \label{figure:parallel-scalability}
\end{figure}

We first study the parallel scalability of the IRLS iterations. The timing results of parallel IRLS iterations on four graphs are plotted in Figure \ref{figure:parallel-scalability}\footnote{The result on euro\_asm is not plotted because its time on 8 cores will dwarf the others.}. In Figure \ref{figure:parallel-scalability}, we see the speedup starts is initially superlinear. This happens because, as $p$ gets large, the total work of applying the block Jacobi preconditioner decreases. On \texttt{asia\_osm}, we can achieve linear speedup until $p=128$. On two of three N-D grid graphs, the linear speedup stops after $p=64$. One reason is that the N-D grid graphs are denser than the road networks (see Table \ref{table:data-set}), which would incur high communication cost when $p$ is large. Another reason is that the use of ILU(0) on N-D grid graphs makes the work reduction not that critical after $p$ is large enough.
\par
In Table \ref{table:irls-time-breakdow} we show the time of different phases of \algname for the given number of cores shown in the last column. We also show in the second to last column the size reduction ratio achieved during the two-level rounding procedure. On the graphs of \texttt{adhead.n26c100} and \texttt{bone.n26c100}, the size reduction is so dramatic that the time taken by the two-level rounding procedure is even less than that of sweep cut on the original graph. However, on the graph of \texttt{liver.n26c100}, the size reduction is not effective and the time of the sequential two-level rounding procedure is even much more than that of IRLS iterations.

\begin{table*}[ht!]
\caption{{Time of different phases of \algname. All the times are in seconds (rounded to integers). The second to last column shows the size reduction ratio of the two-level graph coarsening.}}
\begin{center}
\vspace{-0.6cm}
\begin{tabular}{rllllll}
  \hline
  Graph & Graph Partition & IRLS & Sweep Cut & Two-level & $|\Vs|/|\Vs_c|$  & Number of cores\\
  \hline
  \texttt{asia\_osm} & 5 & 28 & 3 & 16 & 80.5 & 128 \\
  \texttt{euro\_osm} & 7 & 185 & 17 & 109 & 36.2 & 128\\
  \texttt{adhead.n26c100} & 1 & 172 & 6 & 1 & 432.4 & 64\\
  \texttt{bone.n26c100} & 1 & 99 & 4 & 1 & 376.9 & 64\\
  \texttt{liver.n26c100} & 1 & 25 & 2 & 69 & 10.1 & 64\\
  \hline
 \end{tabular}
\end{center}
\vspace{-\baselineskip}
\label{table:irls-time-breakdow}
\end{table*}

\begin{table*}[ht!]
\caption{{Time comparison between \algname and \bksolver. All the times are in seconds (rounded to integers). We use the time of two-level for getting the total time of \algname.}}
\begin{center}
\vspace{-0.6cm}
\begin{tabular}{rllll}
  \hline
  Graph & \algname & \bksolver & Speedup & Number of Cores \\
  \hline
  \texttt{asia\_osm} & 49 & 1465 & 29.8 & 128\\
  \texttt{euro\_osm} & 301 & 3102 & 10.3 & 128 \\
  \texttt{adhead.n26c100} & 174 & 555 & 3.2 & 64 \\
  \texttt{bone.n26c100} & 101 & 215 & 2.1 & 64 \\
  \texttt{liver.n26c100} & 95 & 294 & 3.1 & 64 \\
  \hline
 \end{tabular}
\end{center}
\vspace{-\baselineskip}
\label{table:time-compare-ilrs-bk}
\end{table*}

Finally, we compare the performance of \algname with that of \bksolver. The \bksolver is a sequential code,  and we run it on one Intel Xeon E5-4617 processor. The total time taken by \algname and \bksolver respectively on the five large graphs are shown in Table \ref{table:time-compare-ilrs-bk}. On the road networks, our \algname achieves desirable speedup using 128 cores. Especially on \texttt{asia\_osm}, \algname is almost 30-times faster than \bksolver. In contrast, we find \bksolver to be really efficient on the N-D grid graphs. It is interesting that even though the road networks are planar and sparser, they seem to be much harder than the N-D grid graphs for the \bksolver. 

\begin{table*}
\caption{{Solution quality of sweep cut and two-level. The right two columns are the relative approximation ratio $\delta$.}}
\begin{center}
\vspace{-0.6cm}
\begin{tabular}{rll}
  \hline
  Graph & Sweep Cut & Two-level \\
  \hline
  \texttt{asia\_osm} & $2.1\times 10^{-3}$ & $3.3\times 10^{-5}$ \\
  \texttt{euro\_osm} & $3.2\times 10^{-2}$ & $6.8\times 10^{-5}$ \\
  \texttt{adhead.n26c100} & $1\times 10^{-4}$ & 0 \\
  \texttt{bone.n26c100} & $9.7\times 10^{-4}$ & 0 \\
  \texttt{liver.n26c100} & $4.9\times 10^{-2}$ & $8.8\times 10^{-4}$ \\
  \hline
 \end{tabular}
\end{center}
\label{table:quality-sweepcut-twolevel}
\vspace{-\baselineskip}
\end{table*}

We then use the output of the \bksolver to evaluate the quality of the approximate $\st$ \mincut achieved by sweep cut and the two-level rounding procedure. Specifically, we denote by $\mu^{\ast}$ the $\st$ \mincut value output from \bksolver, and $\mu$ the $\st$ \mincut value output from \algname. Then we compute the relative approximation ratio
\begin{equation}
\delta = (\mu - \mu^{\ast}) / \mu^{\ast}
\end{equation}
As shown in Table \ref{table:quality-sweepcut-twolevel}, the two-level rounding procedure always produces a much better solution than sweep cut, which is a justification of it taking more time on three graphs in Table \ref{table:irls-time-breakdow}.

\section{Discussion}
\label{section:discussions}
The \algname algorithm is a step towards our goal of a scalable, parallel $\st$ \mincut solver for problems with hundreds of millions or billions of edges. It has much to offer, and demonstrates good scalability and improvement over an expertly crafted serial solver.

We have not yet discussed how to take advantage from solving a sequence of related $\st$ \mincut problems because we have not yet completed our investigation into that setting. That said, we have designed our framework such that solving a sequence of problems is reasonably straightforward. Note that all of the set up overhead of \algname can be amortized when solving a sequence of problems if the graph structure and edge weights do not change dramatically. %only the weights on the edges change. 
The larger issue with a sequence of problems, however, is that our solver only produces $\delta$-accurate solutions. For the applications where a sequence is desired, we need to revisit the underlying methods and see if they can adapt to this type of approximate solutions. This setting also offers an opportunity to study the accuracy required for solving a sequence of problems. It is possible that by solving each individual problem to a lower accuracy, it may generate a sequence with a superior final objective value for application purposes. %in the limiting case for a problem such as the GraphCut heuristic.
Thus, the goal of our future investigation on a sequence of problems will attempt to solve the sequence faster and better through carefully engineered approximations.

There are also a few further investigations we plan to conduct regarding our framework. First, we used a distributed memory solver for the diagonally dominant Laplacians. Recent work on the nearly linear time solvers for such systems shows the potential for fast runtime, but the parallelization potential is not yet evident in the literature. We wish to explore the possibility of parallelizing these nearly linear time solvers as well as incorporating them into our framework.
%Nonetheless, we wish to explore possibilities to combine this style of algorithm with our own. For instance, we could easily use such a solver as a preconditioner instead of the ILU(0) method.

Second, we wish to make our two-level rounding procedure more rigorous. As we have mentioned, this method is similar to Lang's flow-based rounding for the spectral method \citep{Lang:2005} that was later formalized in \citep{AndersenL:2008}. %as the FlowImprove algorithm itself.
%We hypothesize that a similar relationship will hold here and we will be able to create a principled approach based on this idea.
We hypothesize that a similar formalization will hold here based on which we will be able to create a more principled approach.
Furthermore, we suspect that a multi-level extension of this will also be necessary as we continue to work on larger and larger problems.

Finally, we plan to continue working on improving the analysis of our \algname algorithm in an attempt to formally bound its runtime. This will involve designing more principled stopping criteria based on tighter diagnostics, e.g., from applying our Cheeger-type inequality.
%akin to our Cheeger-type inequality.

%\fontsize{8}{9}\selectfont
\bibliographystyle{plain}
\bibliography{pirmcut}
%\normalsize
\appendix
\section{Appendix}

\subsection{Proof of the Cheeger-type inequality}
\label{sec:cheeger-proof}
We first prove the following characterization of $\lambda_2$.
\begin{proposition}
$\lambda_2$ is the optimal value of the constrained WLS problem
\begin{align}
\label{eq:lambda2-constrained-wls}
\min_{\x} & \ \ \ \ \frac{1}{2C}\x^{\Tr}\L\x\\
\nonumber
\mbox{s.t.} & \ \ \ \ \x_s = 1, \ \ \ \x_t = -1
\end{align}
\end{proposition}
\begin{proof}
By inspection, $\lambda_1=0$ with $\x = \e$. Thus, $\lambda_2$ can be represented by
\begin{align}
\label{eq:generalized-lambda2-representation}
\lambda_2 &= \min_{\e^{\Tr}\D\x=0}\frac{\x^{\Tr}\L\x}{\x^{\Tr}\D\x}
\end{align}
The constraint $\e^{\Tr}\D\x=0$ is equivalent to $\x_s+\x_t=0$. Because the representation (\ref{eq:generalized-lambda2-representation}) is invariant to scaling of $\x$, we can constrain $\x_s=1$ and $\x_t=-1$, which results in $\x^{\Tr}\D\x=2C$. 
\qed
\end{proof}
Note that (\ref{eq:lambda2-constrained-wls}) has a form similar to (\ref{eq:constrained-wls}), where the main difference is that we use $\{1,-1\}$ to encode the boundary condition. We now prove the Cheeger-type inequality.
\begin{theorem}
$\frac{\phi^2}{2}\leq \lambda_2 \leq 2\phi$
\end{theorem}

\begin{proof}
We first prove the direction of $\lambda_2 \leq 2\phi$. Let $(\Ss,\Sscomp)$ be an $\st$ \mincut such that $s\in\Ss$ and $t\in\Sscomp$. Then $\phi(\Ss)=\phi$. We define the $\{1,-1\}$ encoded cut indicator vector
\begin{align*}
\x_u&=\left\{\begin{array}{ll}
1 & \mbox{ if } u\in\Ss\\
-1 & \mbox{ if } u\in\Sscomp
\end{array}\right.
\end{align*}
This vector satisfies $\e^{\Tr}\D\x=0$. Then by representation (\ref{eq:generalized-lambda2-representation})
\begin{align*}
\lambda_2&\leq \frac{\x^{\Tr}\L\x}{\x^{\Tr}\D\x}=\frac{4 cut(\Ss,\Sscomp)}{2C}\\
&=2\phi(\Ss)=2\phi
\end{align*}
We then prove the direction of $\frac{\phi^2}{2}\leq \lambda_2$. We follow the proof technique and notations used in Theorem 1 of \citep{Chung:2007}. Let $\g$ denote the solution to (\ref{eq:lambda2-constrained-wls}). Using an argument similar to Proposition \ref{prop:interval-constraint-satisfaction}, we can show $-1\leq \g(u)\leq 1$. We sort the nodes according to $\g$ such that
\begin{align}
\label{eq:cheeger-proof-ordering}
1&=\g(s)=\g(v_1)\geq\cdots\geq \g(v_n)=\g(t)=-1
\end{align}
Let $\Ss_i=\{v_1,\ldots,v_i\}$ for $i=1,\ldots,n-1$ and we denote by $\tvol(\Ss_i)=\min\{\vol(\Ss_i),\vol(\Sscomp_i)\}=C$. We let $\alpha=\min_{i=1}^{n-1}\phi(\Ss_i)$. Generalizing the proof to Theorem 1 in \citep{Chung:2007} we can show
\begin{align*}
\lambda_2\geq \frac{\left(\sum_{u\sim v}c(\{u,v\})(\g_{+}(u)^2-\g_{+}(v)^2)\right)^2}{\left(\sum_{u}\g_{+}(u)^2\bd(u)\right)\left(\sum_{u\sim v}c(\{u,v\})(\g_{+}(u)+\g_{+}(v))^2\right)}
\end{align*}
To bound the denominator, we have
\begin{align*}
&\sum_{u\sim v}c(\{u,v\})(\g_{+}(u)+\g_{+}(v))^2\\
\leq& 2\sum_{u\sim v}c(\{u,v\})(\g_{+}(u)^2+\g_{+}(v)^2)\leq 2\g_{+}(s)^2\bd(s)
\end{align*}
where the last inequality follows from $\g_{+}(u)\leq \g_{+}(s)$ and $\bd(s)=C=2\sum_{u\sim v}c(\{u,v\})$. Because $\g_{+}(t)=0$ and according to the definition of (\ref{eq:cheeger-node-weight}) we have $\sum_{u}\g_{+}(u)^2\bd(u)=\g_{+}(s)^2\bd(s)$. By telescoping the term $\g_{+}(u)^2-\g_{+}(v)^2$ using the ordering (\ref{eq:cheeger-proof-ordering}), we then have
\begin{align*}
\lambda_2&\geq \frac{\left(\sum_{i=1}^{n-1}(\g_{+}(v_i)^2-\g_{+}(v_{i+1})^2)cut(\Ss_i,\Sscomp_i)\right)^2}{2(\g_{+}(s)^2\bd(s))^2}\\
&\geq \frac{\alpha^2}{2} \frac{\left(\sum_{i=1}^{n-1}(\g_{+}(v_i)^2-\g_{+}(v_{i+1})^2)\tvol(\Ss_i)\right)^2}{(\g_{+}(s)^2\bd(s))^2}= \frac{\alpha^2}{2}
\end{align*}
where the last equality follows from $\sum_{i=1}^{n-1}(\g_{+}(v_i)^2-\g_{+}(v_{i+1})^2)\tvol(\Ss_i)=\g_{+}(s)^2\bd(s)$ because $\tvol(\Ss_i)=C=\bd(s)$ for $i=1,\ldots,n-1$. From $\alpha\geq \phi$ we have $\lambda_2\geq \frac{\phi^2}{2}$. \qed
\end{proof}

\end{document}